\documentclass[a4paper, 10pt]{article}    

\pdfoutput=1

\usepackage{graphicx}          
\usepackage{microtype}
\usepackage{amsmath}
\usepackage{amssymb}
\usepackage{mathtools}
\usepackage{color,soul}
\usepackage{url}
\usepackage{comment}
\usepackage[charter]{mathdesign}

\usepackage[dvipsnames]{xcolor}
\usepackage{tikz}
\usepackage{pgfplots}
\usepackage{grffile}
\pgfplotsset{compat=newest}
\usetikzlibrary{automata}
\usetikzlibrary{arrows, arrows.meta, decorations.markings, patterns}
\usepgfplotslibrary{patchplots, groupplots}
\usetikzlibrary{shapes, positioning, fit, backgrounds}
\usepgfplotslibrary{fillbetween}

\usetikzlibrary{overlay-beamer-styles}
\usetikzlibrary{calc}

\usepackage{calc}

\definecolor{tudcyan}{RGB}{0,166,214}
\definecolor{tudmagenta}{RGB}{109,23,127}
\definecolor{tudpurple}{RGB}{29,28,115}
\definecolor{tudgraygreen}{RGB}{107,134,137}

\colorlet{lighttudcyan}{tudcyan!20}
\colorlet{lighttudmagenta}{tudmagenta!20}

\newlength{\hatchspread}
\newlength{\hatchthickness}
\newlength{\hatchshift}
\newcommand{\hatchcolor}{}
\tikzset{hatchspread/.code={\setlength{\hatchspread}{#1}},
	hatchthickness/.code={\setlength{\hatchthickness}{#1}},
	hatchshift/.code={\setlength{\hatchshift}{#1}},
	hatchcolor/.code={\renewcommand{\hatchcolor}{#1}}}
\tikzset{hatchspread=7pt,
	hatchthickness=0.5pt,
	hatchshift=0pt,
	hatchcolor=black}
\pgfdeclarepatternformonly[\hatchspread,\hatchthickness,\hatchshift,\hatchcolor]
{custom north west lines}
{\pgfqpoint{\dimexpr-2\hatchthickness}{\dimexpr-2\hatchthickness}}
{\pgfqpoint{\dimexpr\hatchspread+2\hatchthickness}{\dimexpr\hatchspread+2\hatchthickness}}
{\pgfqpoint{\dimexpr\hatchspread}{\dimexpr\hatchspread}}
{
	\pgfsetlinewidth{\hatchthickness}
	\pgfpathmoveto{\pgfqpoint{0pt}{\dimexpr\hatchspread+\hatchshift}}
	\pgfpathlineto{\pgfqpoint{\dimexpr\hatchspread+0.15pt+\hatchshift}{-0.15pt}}
	\ifdim \hatchshift > 0pt
	\pgfpathmoveto{\pgfqpoint{0pt}{\hatchshift}}
	\pgfpathlineto{\pgfqpoint{\dimexpr0.15pt+\hatchshift}{-0.15pt}}
	\fi
	\pgfsetstrokecolor{\hatchcolor}
	\pgfusepath{stroke}
}

\def\centerarc[#1](#2)(#3:#4:#5)
{ \draw[#1] ($(#2)+({#5*cos(#3)},{#5*sin(#3)})$) arc (#3:#4:#5); }

\newcommand{\rev}[1]{{#1}}

\newtheorem{prob}{Problem}
\newtheorem{assum}{Assumption}
\newtheorem{defn}{Definition}
\newtheorem{rem}{Remark}
\newtheorem{prop}{Proposition}
\newtheorem{lem}{Lemma}
\newtheorem{thm}{Theorem}

\def\proofname{Proof}

\def\QED{\mbox{\rule[0pt]{1.3ex}{1.3ex}}} 

\newif\if@QEDshow  \@QEDshowtrue
\def\proofindentspace{2\parindent}
\def\proof{\@ifnextchar[{\@proof}{\@proof[\proofname]}}
\def\@proof[#1]{\@QEDshowtrue\par\noindent\hspace{\proofindentspace}{\itshape #1: }}

\def\QEDhereeqn{\global\@QEDshowfalse\eqno\let\eqno\relax\let\leqno\relax
	\let\veqno\relax\hbox{\QED}}
\def\QEDhere{\global\@QEDshowfalse\QED}
\def\QEDoff{\global\@QEDshowfalse}

\newcommand*{\tran}{^{\mkern-1.5mu\mathsf{T}}\!}  

\def\d{\ensuremath{\mathrm{d}}}

\DeclareMathOperator{\Post}{Post}
\DeclarePairedDelimiter{\ceil}{\lceil}{\rceil}
\DeclareMathOperator{\bigO}{\mathcal{O}}
\def\norm[#1]{\left|#1\right|}
\def\shortnorm[#1]{|#1|}
\def\bisim{\ensuremath{\cong}}

\def\plant{{\mathrm{p}}}

\def\periodic{{\mathrm{P}}}
\def\etc{{\mathrm{PETC}}}
\def\Bisim{{\mathrm{B}}}
\def\Sim{{\mathrm{S}}}
\def\MPETC{{\mathrm{E}}}
\def\PETC{{\mathrm{S'}}}
\def\dummy{}

\def\Xs{\mathcal{X}}
\def\Ys{\mathcal{Y}}
\def\Us{\mathcal{U}}

\def\No{\mathbb{N}_{0}}
\def\N{\mathbb{N}}
\def\R{\mathbb{R}}
\def\S{\mathbb{S}}

\def\Rs{\mathcal{R}}

\def\Ss{\mathcal{S}}
\def\Ks{\mathcal{K}}
\def\Qs{\mathcal{Q}}

\def\KN{\mathbb{K}_N}

\def\Es{\mathcal{E}}  
\def\Us{\mathcal{U}}  

\def\np{{n_\plant}}
\def\nup{{n_{\mathrm{u}}}}
\def\nx{{n_{\mathrm{x}}}}

\def\kmax{\bar{k}}
\def\kmin{\underline{k}}

\def\xv{\boldsymbol{x}}

\def\xiv{\boldsymbol{\xi}}

\def\zetav{\boldsymbol{\zeta}}

\def\Am{\boldsymbol{A}}
\def\Bm{\boldsymbol{B}}

\def\I{\mathbf{I}}
\def\O{\mathbf{0}}

\def\Mm{\boldsymbol{M}}
\def\Nm{\boldsymbol{N}}

\def\Pm{\boldsymbol{P}}
\def\Qm{\boldsymbol{Q}}

\def\Km{\boldsymbol{K}}

\def\Pl{\Pm_{\mathrm{Lyap}}}
\def\Ql{\Qm_{\mathrm{Lyap}}}

\def\e{\mathrm{e}}

\title{Towards Traffic Bisimulation of Linear Periodic Event-Triggered Controllers%
	\thanks{© 2020 IEEE.  Personal use of this material is permitted.  Permission from IEEE must be obtained for all other uses, in any current or future media, including reprinting/republishing this material for advertising or promotional purposes, creating new collective works, for resale or redistribution to servers or lists, or reuse of any copyrighted component of this work in other works.}}

\author{Gabriel de A.~Gleizer and Manuel Mazo Jr.
	\thanks{This work is supported by the European Research Council through the SENTIENT project (ERC-2017-STG \#755953).	G.~de A.~Gleizer and M.~Mazo Jr. are with the Delft Center for Systems and Control,
		Delft Technical University, 2628 CD Delft, The Netherlands
		{\tt\small \{g.gleizer, m.mazo\}@tudelft.nl}}%
}

\begin{document}
	
	\maketitle
	\begin{abstract}
		We provide a method to construct finite abstractions exactly bisimilar to linear systems under a modified periodic event-triggered control (PETC), when considering as output the inter-event times they generate. Assuming that the initial state lies on a known compact set, these finite-state models can exactly predict all sequences of sampling times until a specified Lyapunov sublevel set is reached.
		\rev{Based on these results,} we provide a way to build tight models simulating the traffic of conventional PETC. These models allow computing tight bounds of the PETC average frequency and global exponential stability (GES) decay rate.
		Our results are demonstrated through a numerical case study.
	\end{abstract}
		
	\section{Introduction}
	Since the beginning of the digital control era, periodic sampling has been the standard choice for practitioners on all sorts of applications, due to its simple implementation and the existence of many analysis and design results and tools. However, with the replacement of point-to-point communication channels with networked control systems (NCSs), wireless networks in particular, minimizing control data generation and transmission becomes crucial. Because of this and the fundamental limitations of periodic control, aperiodic methods such as event-triggered control (ETC) have been proposed \cite{tabuada2007event} and gained enormous traction since then.
	In ETC, instead of time determining when a sensor should send the data, this is dictated by a triggering condition, typically a function of current and held measurements. From the beginning, ETC has shown immense promise in reducing control communications, and subsequent work was dedicated, among other objectives, to further decrease the number of transmissions generated \cite{wang2008event, girard2015dynamic, heemels2012introduction}. While most event-triggered mechanisms assume that a smart sensory system continuously monitors the designed triggering condition, periodic event-triggered control (PETC, \cite{heemels2013periodic}) considers the case where such sensory system is also digital, and the condition checking is periodic. For this practical reason, this paper focuses on this class of ETC.
	
	Even though numerical simulations provided in much of the literature provide evidence that the promised traffic reduction is substantial \cite{wang2008event, heemels2012introduction, heemels2013periodic}, little is formally known about the traffic patterns generated by (P)ETC. Existing work can be divided into two categories: the first involves creating abstract models of the generated traffic, such as \cite{kolarijani2016formal, mazo2018abstracted, fu2018traffic, gleizer2020scalable}; the second involves inferring asymptotic properties of the trajectories that inter-event times describe \cite{postoyan2019interevent}. In the latter, planar linear systems using the triggering condition from \cite{tabuada2007event} are investigated, and the authors show that, under some conditions on the triggering function, inter-event times either converge to a constant value or to a periodic pattern. Our ultimate goal, however, is designing tools for scheduling, and thus we are more concerned with precise short-term predictions than longer-term tendencies. Consequently, the current piece of work \rev{belongs to} the first category. The work in \cite{kolarijani2016formal, mazo2018abstracted} is dedicated to approximate similar models \cite{tabuada2009verification} of (continuous) ETC traffic using the novel notion of power quotient systems: in this case, the abstraction's states represent regions of the state-space of the ETC system, and the output associated to a discrete state of the abstraction is a time interval, to which the actual inter-event time is guaranteed to belong. This modeling strategy is extended to PETC in \cite{fu2018traffic}. To obtain an exact simulation, \cite{gleizer2020scalable} explores the discrete nature of inter-event times of PETC, leading to a novel quotient system that predicts the exact inter-sample time from any given state. Additionally, the main advantage of the latter in comparison to its predecessors is scalability w.r.t.~state-space dimension. On the downside, the generated models exhibit severe non-determinism, likely due to the small number of states of the abstraction and the relaxations used when computing the transition relation. Therefore, even though its one-step ahead predictions are exact, after a couple of steps it loses its prediction capability is severely limited.
	
	The present work tackles precisely this longer-term predictability issue. Building upon the quotient model from \cite{gleizer2020scalable}, we develop new abstractions that enumerate all possible sequences of inter-event times until \rev{a Lyapunov sublevel set is reached. Based on this, we propose a modified PETC mixed with periodic control, hereafter denoted MPETC, that initiates with PETC sampling and switches to periodic sampling when the states lie inside the aforementioned sublevel set. The MPETC retains the practical benefits from PETC, while improving traffic predictability; for it, our abstraction constitutes a bisimulation.}
	The abstraction is computable because the number of possible sampling sequences generated during the PETC \rev{phase} is finite, and checking whether PETC generates a given sampling sequence is decidable: it is equivalent to the satisfiability of a conjunction of non-convex quadratic inequalities, to which solvers exist (e.g.,~\rev{the satisfiability-modulo-theories (SMT) solver }Z3~\cite{demoura2008z3}). In our symbolic model, each state is associated with a sequence of inter-event times, which is similar in spirit to \cite{lecorronc2013mode}. When generating these finitely many discrete states, exhaustive search can be avoided by employing a recursive algorithm.
	
	The clear advantages of our new model are the exact enumeration of sampling sequences that can be generated by PETC on a significant future horizon and the establishment of tight bounds on the Lyapunov function convergence speed. Naturally, a disadvantage of our presented method is a substantial growth in the number of discrete states when compared to \cite{gleizer2020scalable}.
	Finally, we show how to modify our bisimilar model to obtain a tight traffic model simulating an unmodified PETC system, \rev{presenting two derived results: tighter decay rate estimation (compared to, e.g., \cite{heemels2013periodic}), and maximum average triggering frequency computation.}

	\subsection{Notation}
	
	We denote by $\No$ the set of natural numbers including zero, $\N \coloneqq \No \setminus \{0\}$, $\N_{\leq n} \coloneqq \{1,2,...,n\}$, and $\R_+$ the set of non-negative reals. We denote by $\ceil{a}$ the smallest integer not smaller than $a\in\R$. We denote by $|\xv|$ the norm of a vector $\xv \in \R^n$, but if $s$ is a sequence or set, $|s|$ denotes its length or cardinality, respectively. For a square matrix $\Am \in \R^{n \times n},$ we write $\Am \succ \O$ ($\Am \succeq \O$) if $\Am$ is positive definite (semi-definite). The set $\S^n$ denotes the set of symmetric matrices in $\R^n$. For $\Pm \in \S^n$, $\lambda_{\max}(\Pm)$ ($\lambda_{\min}(\Pm)$) denotes its maximum (minimum) eigenvalue. \rev{For a set $\Xs\subseteq\Omega$, we denote by $\bar{\Xs}$ its complement: $\Omega \setminus \Xs$}.
	\rev{For a relation $\Rs \subseteq \Xs_a \times \Xs_b$, its inverse is denoted as} $\Rs^{-1} = \{(x_b, x_a) \in \Xs_b \times \Xs_a : (x_a, x_b) \in \Rs\}$. Finally, we denote by $\pi_\Rs(\Xs_a) \coloneqq \{ x_b \in \Xs_b | (x_a, x_b) \in \Rs \text{ for some } x_a \in \Xs_b\}$ the natural projection $\Xs_a$ onto $\Xs_b$.
	
	We say that an autonomous system $\dot{\xiv}(t) = f(\xiv(t))$ is globally exponentially stable (GES) if there exist $M < \infty$ and $b > 0$ such that every solution of the system satisfies $|\xiv(t)| \leq M\e^{-bt}|\xiv(0)|$ for every initial state $\xiv(0)$; moreover, we call $b$ a decay rate estimate of the system. When convenient, we use $\xiv_{\xv}(t)$ to denote a trajectory from initial state $\xiv(0) = \xv.$
	
	\section{Preliminaries}
	
	\subsection{Periodic event-triggered control}
	
	Consider a linear plant controlled with sample-and-hold state feedback described by
	\begin{align}
	\dot{\xiv}(t) &= \Am\xiv(t) + \Bm\Km\hat{\xiv}(t),\label{eq:plant}\\
	\xiv(0) &= \hat{\xiv}(0) = \rev{\xv_0}, \nonumber
	\end{align}
	where $\xiv(t) \in \R^\nx$ is the plant's state with initial value $\rev{\xv_0}$, $\hat{\xiv}(t) \in \R^\nx$ is the measurement of the state available to the controller, $\Km\hat{\xiv}(t) \in \R^\nup$ is the control input, $\nx$ and $\nup$ are the state-space and input-space dimensions, respectively, and $\Am, \Bm, \Km$ are matrices of appropriate dimensions.
	The holding mechanism is zero-order: let $t_i \in \R_+, i \in \N_0$ be a sequence of sampling times, with $t_0 = 0$ and $t_{i+1} - t_i > \varepsilon$ for some $\varepsilon > 0$; then $\hat{\xiv}(t) = \xiv(t_i), \forall t \in [t_i, t_{i+1})$.
	
	In ETC, a \emph{triggering condition} determines the sequence of times $t_i$. In PETC, this condition is checked only periodically, with a fundamental checking period $h$. We consider the family of static \emph{quadratic triggering conditions} from \cite{heemels2013periodic} with an additional maximum inter-event time condition below:
	\begin{equation}\label{eq:quadtrig}
	t_{i+1} = \inf\left\{\!kh>t_i, k \in \N \middle|
	\!\begin{array}{c}
	\begin{bmatrix}\xiv(kh) \\ \rev{\xiv(t_i)}\end{bmatrix}\tran
	\!\Qm \begin{bmatrix}\xiv(kh) \\ \rev{\xiv(t_i)}\end{bmatrix} > 0\! \\
	\rev{\text{ or }} \ kh-t_i \leq \bar{k}h\phantom{\dot{\hat{I}}}
	\end{array}\!
	\right\}\!,
	\end{equation}
	where $\Qm \in \S^{2\nx}$ is the designed triggering matrix, and $\bar{k}$ is the chosen maximum (discrete) inter-event time.\footnote{Often a maximum inter-event time arises naturally from the closed-loop system itself (see \cite{gleizer2018selftriggered}). Still, one may want to set a smaller maximum inter-event time so as to establish a ``heart beat'' of the system.} Many of the triggering conditions available in the literature can be written as in Eq.~\eqref{eq:quadtrig}; the interested reader may refer to \cite{heemels2013periodic} for a comprehensive list of triggering and stability conditions.
	
	As noted in \cite{gleizer2020scalable}, the inter-event time $t_{i+1} - t_{i}$ is a function of \rev{$\xv_i \coloneqq \xiv(t_i)$}; denoting $\kappa \coloneqq (t_{i+1}-t_i)/h$ as the discrete inter-event time, from Eq.~\eqref{eq:quadtrig} it follows that
	\begin{gather}
	\kappa(\rev{\xv_i}) = \min\left\{k \in \{1, 2, ...\bar{k}\}: \rev{\xv_i}\tran\Nm(k)\rev{\xv_i} > 0 \rev{\text{ or }} k=\bar{k}\right\}, \nonumber\\
	\Nm(k) \coloneqq \begin{bmatrix}\Mm(k) \\ \I\end{bmatrix}\tran
	\Qm \begin{bmatrix}\Mm(k) \\ \I\end{bmatrix}, \label{eq:petc_time}\\
	\!\!\Mm(k) \coloneqq \Am_\d(k) + \Bm_\d(k)\Km \coloneqq \e^{\Am hk} + \int_0^{hk}\e^{\Am\tau}\d\tau \Bm\Km.\!\!\nonumber
	\end{gather}
	where $\I$ denotes the identity matrix.
	
	\subsection{Transition systems}
	
	We use the framework of \cite{tabuada2009verification} to formally relate systems of different natures, e.g., those described by differential equations with those described by finite-state machines. First, a generalized notion of transition system is given:
	\begin{defn}[Transition System \cite{tabuada2009verification}]\label{def:system} 
		A system $\Ss$ is a tuple $(\Xs,\Xs_0,\Us,\Es,\Ys,H)$ where:
		\begin{itemize}
			\item $\Xs$ is the set of states,
			\item $\Xs_0 \subseteq \Xs$ is the set of initial states,
			\item $\Us$ is the set of inputs,
			\item $\Es \subseteq \Xs \times \Us \times \Xs$ is the set of edges (or transitions),
			\item $\Ys$ is the set of outputs, and
			\item $H: \Xs \to \Ys$ is the output map.
		\end{itemize}
	\end{defn}
	A system is said to be finite (infinite) state when the cardinality of $\Xs$ is finite (infinite). A system is autonomous if $\Us = \emptyset$, in which case a transition is denoted by a pair $(x, x') \in \Xs\times\Xs$ instead of a triplet. Hereafter, we focus on autonomous systems. For these cases, we define $\Post_\Ss(x) \coloneqq \{x'| (x,x') \in \Es\}$. 
	We call $x_0 \to x_1 \to x_2 \to ...$ an \emph{infinite internal behavior} of $\Ss$ if $x_0 \in \Xs_0$ and $(x_i,x_{i+1}) \in \Es$, and $y_0 \to y_1 \to ...$ its corresponding \emph{infinite external behavior}, or \emph{trace}, if $H(x_i) = y_i$. 
	
	The concept of simulation relation is fundamental for relating two transition systems.
	
	\begin{defn}[Simulation Relation \cite{tabuada2009verification}]\label{def:sim}
		Consider two systems $\Ss_a$ and $\Ss_b$ with $\Ys_a$ = $\Ys_b$. A relation $\Rs \subseteq \Xs_a \times \Xs_b$ is a simulation relation from $\Ss_a$ to $\Ss_b$ if the following conditions are satisfied:
		\begin{enumerate}
			\item[i)] for every $x_{a0} \in \Xs_{a0}$, there exists $x_{b0} \in \Xs_{b0}$ with $(x_{a0}, x_{b0}) \in \Rs;$
			\item[ii)] for every $(x_a, x_b) \in \Rs, H_a(x_a) = H_b(x_b);$
			\item[iii)] for every $(x_a, x_b) \in \Rs,$ we have that $(x_a, x_a') \in \Es_a$ implies the existence of $(x_b, x_b') \in \Es_b$ satisfying $(x_a', x_b') \in \Rs.$
		\end{enumerate}
	\end{defn}
	We say $\Ss_a \preceq \Ss_b$ when $\Ss_b$ simulates $\Ss_a$. 
	If $\Ss_a \preceq \Ss_b$, from the definition above, it becomes clear that any sequence of outputs from $\Ss_a$ can be generated by $\Ss_b$; the converse is not true, unless there is in fact a bisimulation:
	\begin{defn}[Bisimulation \cite{tabuada2009verification}]
		Consider two systems $\Ss_a$ and $\Ss_b$ with $\Ys_a$ = $\Ys_b$.
		$\Ss_a$ is said to be bisimilar to $\Ss_b$, denoted $\Ss_a \bisim \Ss_b$, if there exists a relation
		$\Rs$ such that:
		\begin{itemize}
			\item $\Rs$ is a simulation relation from $\Ss_a$ to $\Ss_b$;
			\item $\Rs^{-1}$ is a simulation relation from $\Ss_b$ to $\Ss_a$.
		\end{itemize}
	\end{defn}
	%
	\section{Problem formulation}
	
	Consider system \eqref{eq:plant}--\eqref{eq:quadtrig}, a quadratic Lyapunov function $V(\xv) = \xv\tran\Pm\xv, \Pm \succ \O$, and the following assumptions:
	\begin{assum}\label{assum:stable} System \eqref{eq:plant}--\eqref{eq:quadtrig} is GES, and there exists a known constant $0 \leq a < 1$ such that every solution of the system satisfies $V(\xiv(t_{i+1})) \leq aV(\xiv(t_i))$.
	\end{assum}
	\begin{rem}\label{rem:a_computation}
	To compute $a$, one can verify the implication $$\forall \xv \!\in\! \R^\nx (\forall i \in \N_{< k} \, \xv\tran\Nm(i)\xv \leq 0) \wedge (\xv\tran\Nm(k)\xv > 0) \implies V(\Mm(k)\xv) \leq a V(\xv)$$ for every $k\in\{1,...\bar{k}\}$. This can be cast as a set of LMIs through the S-procedure.
	\end{rem}
	
	Note that, from Eq.~\eqref{eq:petc_time}, no triggering can occur at $k$ if $\Nm(k) \preceq 0$. Thus, we can determine the global minimum inter-sample time as $h\underline{k}$, with $\underline{k} \!\coloneqq\! \min\{k\! \in\! \{1,...,\bar{k}\}: \Nm(k)\! \npreceq\! \O\}$.
	
	\begin{assum}\label{assum:periodicstable} For system \eqref{eq:plant}, there exists some $h_{\periodic} > 0$ such that the periodic sampling sequence with $t_{i+1} = t_i + h_{\periodic}$ ensures $V(\xiv(t_{i+1}) \leq V(\xiv(t_i))$.
	\end{assum}
	
	This will not necessarily follow from Assumption \ref{assum:stable}; however, ETC is typically designed based on a continuous-time Lyapunov function, and for small enough values of $h$, the same Lyapunov function will work for periodic control.%
	\footnote{This is easy to see when one considers the first order approximation of the discrete-time transition matrix $\e^{\Am h} \approxeq \I + \Am h$. If the continuous-time Lyapunov inequality $\Am\tran\Pm + \Pm\Am \preceq -\epsilon\I$ holds for some $\epsilon > 0$ then:  $h(\Am\tran\Pm + \Pm\Am) \preceq -h\epsilon\I 
		\iff (\I + \Am h)\tran\Pm(\I + \Am h) - \Pm \preceq -h\epsilon\I + h^2\Am\tran\Pm\Am$, which for sufficiently small $h$ results in $\e^{\Am\tran h}\Pm\e^{\Am h} - \Pm \preceq -h\epsilon\I \prec \O$, i.e. the discrete-time Lyapunov inequality.}
	
	\begin{assum}\label{assum:compact} A value $V_0 > 0$ is known such that $\xiv(0) \in \Xs_0 = \{\xv\in\R^\np: V(\xv) \leq V_0\}$.
	\end{assum}
	
	Now let us propose a modification to the PETC system. \rev{Since ETC can reduce communication frequency while ensuring a fast decay rate, it makes practical sense to focus on ETC during the transient phase. However, once states are close enough to the origin, decay rates have disputable practical relevance. Therefore, we admit that, when $\hat{\xiv}(t)$ enters a small sublevel set $\Xs_{\periodic} \coloneqq \{\xv\in\R^\np| V(x) \leq rV_0\}, r < 1$, the controller can switch to periodic sampling, \rev{with $h_{\periodic}$ significantly bigger than $h$}; in fact, it can be as big as possible, provided that it preserves Assumption \ref{assum:periodicstable}. This results in more predictable (hence schedulable) traffic while retaining a reduction of traffic.}  
	Let us denote by $t_{i+1}(t_i, \xiv(t_i))|_{\mathrm{PETC}}$ the solution of Eq.~\eqref{eq:quadtrig}. Mathematically, the mixed sampling strategy\rev{, hereafter denoted MPETC,} dictates the sampling times as follows:
	\begin{equation}\label{eq:mixedtiming}
	\begin{aligned}
	t_{i+1} &= t_{i+1}(t_i, \xiv(t_i))|_{\mathrm{PETC}}, && V(\xiv(t_i)) > rV_0 \\
	t_{i+1} &= t_i + h_{\periodic}, && V(\xiv(t_i)) \leq rV_0.
	\end{aligned}
	\end{equation}
	
	Hereafter, denote $\Xs_{\periodic} \coloneqq \{\xv \in \R^\nx|V(\xv) \leq rV_0\} = r\Xs_0$. 
	This system has the following infinite-state traffic model: $\Ss_{\MPETC} \coloneqq (\Xs, \Xs_0, \emptyset, \Es_\etc \cup \Es_\periodic, \Ys_{\MPETC}, H_{\MPETC})$ where
	\begin{multline}\label{eq:original}
	\begin{aligned}
	\Xs &= \Xs_0; \\
	\Es_\etc &= \{(\xv, \xv') \in (\Xs \setminus \Xs_{\periodic}) \times \Xs: \xv' = \xiv_{\xv}(h\kappa(\xv))\}; \\
	\Es_\periodic &= \{(\xv, \xv') \in \Xs_{\periodic} \times \Xs_{\periodic}: \xv' = \xiv_{\xv}(h_{\periodic})\}; \\
	\Ys_{\MPETC} &= \{h, 2h, ..., \bar{k}h, h_{\periodic}\};\\
	H_{\MPETC}(\xv) &= \begin{cases} h\kappa(x), & \xv \in \Xs \setminus \Xs_{\periodic}, \\
	h_{\periodic}, & \xv \in \Xs_{\periodic}.
	\end{cases}
	\end{aligned}
	\raisetag{3\baselineskip}
	\end{multline}
	For states starting outside $\Xs_{\periodic}$, transitions and outputs (the inter-sample times) are dictated by the PETC strategy; for states inside $\Xs_{\periodic}$, transitions and outputs are dictated by periodic sampling. Note that $\Es_\periodic$ is defined over $\Xs_{\periodic} \times \Xs_{\periodic}$, i.e., states starting in $\Xs_{\periodic}$ always land in $\Xs_{\periodic}$: this comes from the fact that the periodic \rev{phase} is forward-invariant due to Assumption \ref{assum:periodicstable}. 
	We are ready to define our main problem:
	\begin{prob}\label{prob:bisim} Considering Assumptions \ref{assum:stable}--\ref{assum:compact}, determine if $\Ss_{\MPETC}$ admits a computable finite-state bisimulation. If so, provide an algorithm to compute it.
	\end{prob}
	
	\section{Main result}
	
	To build a bisimilar model of $\Ss_{\MPETC}$, the main observation is that eventually all trajectories of the system \eqref{eq:plant}, \eqref{eq:mixedtiming} enter $\Xs_{\periodic}$, which follows from Assumption \ref{assum:stable}. Clearly, when in $\Xs_{\periodic}$ the system admits a trivial, single-state traffic bisimulation:
	\begin{prop}\label{prop:periodic_bisim}
		Define $H_\periodic: \R^\nx \to \R$ such that $H_\periodic \equiv h_{\periodic}$. The system $$\Ss_\periodic^\Bisim = (\{\Xs_{\periodic}\}, \{\Xs_{\periodic}\}, \emptyset, \{(\Xs_{\periodic}, \Xs_{\periodic})\}, \{h_{\periodic}\}, H_\periodic)$$ is a bisimilar quotient system of $(\Xs_{\periodic}, \Xs_{\periodic}, \emptyset, \Es_\etc \cup \Es_\periodic, \Ys_{\MPETC}, H_{\MPETC})$.
	\end{prop}
	
	Another important observation is that, since the PETC \rev{phase} is asymptotically stable (Assumption \ref{assum:stable}), states from $\Xs_0$ reach $\Xs_{\periodic}$ in finite time. Thus, for any state in $\Xs_0$, there is a finite number of PETC-generated samples, after which all samples are periodically taken. Let $K \coloneqq \{\kmin, \kmin+1,...,\kmax\}$; since at each step there are finitely many ($|K|$) inter-sample time possibilities, we can state the following:
	\begin{lem}\label{lem:finite} Let Assumptions \ref{assum:stable}--\ref{assum:compact} hold, define $N \coloneqq \ceil{\log_a(r)}$. Then system \eqref{eq:original} can produce at most $|K|(\left(|K|-1\right)^N-1)(|K|-1)^{-1}$ different traces.
	\end{lem}
	\begin{proof}
		Using assumption \ref{assum:stable}, recursively apply $V(\xiv(t_{i+1}) \leq aV(\xiv(t_i))$ to get $V(\xiv(t_N) \leq a^NV(\rev{\xv_0}) \leq a^NV_0.$ Then, $N > \log_a(r)$ implies $a^NV_0 \leq rV_0$; thus, it takes at most $N$ steps to enter $\Xs_{\periodic}$. After this, from Proposition \ref{prop:periodic_bisim}, the remaining trace is $h_{\periodic},h_{\periodic},...$. This is the trace if $\rev{\xv_0}\in\Xs_{\periodic}$, which accounts for one trace; $\Ss_{\MPETC}$ has at most $|K|$ traces for which it takes one step to reach $\Xs_{\periodic}$ from $\rev{\xv_0}$, at most $|K|^2$ traces for which it takes two steps to reach $\Xs_{\periodic}$, and etc., up to $|K|^N$ for the maximum number of steps. Summing up this geometric series gives $|K|(\left(|K|-1\right)^N-1)(|K|-1)^{-1}$.
	\end{proof}
	
	Lemma \ref{lem:finite} permits the construction of a rather straightforward finite-state model similar to $\Ss_{\MPETC}$.  Denote by $K^m$ the set of all sequences of length $m$ of the form $(k_i)_{i=0}^m, k_i \in K.$ We create one state for each sequence in $K^m$. The state $\dummy{k_1k_2...k_m}$ is associated with the trace $hk_1, hk_2, ..., hk_m, h_{\periodic}, h_{\periodic}, ...$, thus taking $m$ samples to enter the periodic \rev{phase}. By definition, its successor must be $\dummy{k_2...k_m}$. Finally, let $\varepsilon$ denote the empty sequence; a state $\dummy{k}$ generates the trace $hk, h_{\periodic}, h_{\periodic}, ...$, and thus its successor is $\dummy\varepsilon$, associated with the periodic \rev{phase}. Hence, $\Post(\dummy\varepsilon) = \dummy\varepsilon$ and $H^{\Sim}(\dummy\varepsilon) = h_{\periodic}$. Let $\KN \coloneqq \cup_{i=1}^NK^i \cup \{\varepsilon\}$; we consolidate this modeling strategy with the following result:
	\begin{prop}\label{prop:simtrivial}
		Let Assumptions \ref{assum:stable}--\ref{assum:compact} hold and $N \coloneqq \ceil{\log_a(r)}$. Consider $\Ss^{\Sim} \!\coloneqq\! \left(\KN, \KN, \emptyset, \Es^{\Sim}, \Ys_{\MPETC}, H^{\Sim}\right)\!$ with
		\begin{itemize}
			\item $\Es^{\Sim} = \{(\dummy{k\sigma}, \dummy\sigma)|\dummy{k\sigma}\in\KN\} \cup \{(\dummy\varepsilon, \dummy\varepsilon)\}$;
			\item $H^{\Sim}(\dummy{k\sigma}) = hk$ and $H^{\Sim}(\dummy\varepsilon) = h_{\periodic}$.
		\end{itemize}
		Then $\Ss^{\Sim} \succeq \Ss_{\MPETC}$. %
	\end{prop}
	
	\begin{proof}
		System $\Ss^{\Sim}$ generates all possible traces of type $hk_1, hk_2, ..., hk_m,\allowbreak h^*, h^*, ...,$ for $0 \leq m \leq N$, which, according to Lemma \ref{lem:finite}, include all possible traces of $\Ss_{\MPETC}$; thus, the behavior of $\Ss^{\Sim}$ contains that of $\Ss_{\MPETC}$. Because both systems $\Ss_{\MPETC}$ and $\Ss^{\Sim}$ are deterministic and non-blocking, this implies that $\Ss^{\Sim} \succeq \Ss_{\MPETC}$ \cite[Proposition 4.11]{tabuada2009verification}.  
	\end{proof}
	
	The set $\KN$ includes sequences that may not be generated by the PETC \rev{phase}. To trim off these spurious sequences, let us define the following relation:
	\begin{defn}[Inter-sample sequence relation]\label{def:bisimrel} We denote by $\Rs_{\Bisim} \subseteq \Xs \times \KN$ the relation satisfying 
		\begin{equation}\label{eq:Xempty}
		(\xv,\dummy\varepsilon) \in \Rs_{\Bisim} \text{ iff } \xv \in \Xs_{\periodic},
		\end{equation}
		and
		$(\xv,\dummy{k_1k_2...k_m}) \in \Rs_{\Bisim}$ if and only if
		\begin{subequations}\label{eq:sequence}
			\begin{align}
			\xv &\in \Xs_0, \label{eq:sequence_x0}\\
			\xv &\in \Qs_{k_1}, \label{eq:sequence_k1}\\
			\Mm(k_1)\xv &\in \Qs_{k_2}, \label{eq:sequence_k2}\\
			\Mm(k_2)\Mm(k_1)\xv &\in \Qs_{k_3}, \label{eq:sequence_k3}\\
			& \vdots \nonumber\\
			\Mm(k_{m-1})...\Mm(k_1)\xv &\in \Qs_{k_m}, \label{eq:sequence_km-1}\\
			\xv &\notin \Xs_{\periodic}, \label{eq:sequence_x0_not_periodic}\\
			\Mm(k_1)\xv &\notin \Xs_{\periodic}, \label{eq:sequence_x1_not_periodic}\\
			& \vdots \nonumber\\
			\Mm(k_{m-1})...\Mm(k_1)\xv &\notin \Xs_{\periodic}, \label{eq:sequence_xm-1_not_periodic} \\
			\Mm(k_m)...\Mm(k_1)\xv &\in \Xs_{\periodic}, \label{eq:sequence_km}
			\end{align}
		\end{subequations}
		where 
		\begin{equation}\label{eq:setq}
		\begin{gathered}
		\Qs_k \coloneqq \Ks_k \setminus \left(\bigcap_{j=\underline{k}}^{k-1} \Ks_{j}\right) = \Ks_k \cap \bigcap_{j=\underline{k}}^{k-1} \bar{\Ks}_{j}, \\
		\Ks_k \coloneqq \begin{cases}
		\{\xv \in \Xs| \xv\tran\Nm(k)\xv > 0\}, & k < \bar{k}, \\
		\Xs, & k = \bar{k}.
		\end{cases}
		\end{gathered}
		\end{equation}
	\end{defn}
	
	Eq.~\eqref{eq:setq}, taken from \cite{gleizer2020scalable}, defines the sets $\Qs_k$, containing the states that trigger exactly with inter-sample time $hk$. Eq.~\eqref{eq:Xempty} determines that states $\xv \in \Xs_\periodic$ are related to the state $\varepsilon$. Finally, a state $\xv \in \R^n$ is related to a state $k_1k_2...k_m$ of the abstraction if the following are satisfied: 1) it belongs to the compact set of interest (Eq.~\eqref{eq:sequence_x0}), 2) the inter-sample time sequence that it generates up until it enters $\Xs_\periodic$ is $hk_1,hk_2,...,kh_m$ (Eqs.~\eqref{eq:sequence_k1}--\eqref{eq:sequence_km-1}), and 3) the sampled states $\xiv_{\xv}(k_1h), \xiv_{\xv}((k_1+k_2)h), ...$ of the trajectory starting from $\xv$ do not belong to $\Xs_\periodic$ (Eqs.~\eqref{eq:sequence_x0_not_periodic}--\eqref{eq:sequence_xm-1_not_periodic}), while the m-th sampled state does belong to $\Xs_\periodic$ (Eq.~\eqref{eq:sequence_km}).
	
	We now employ the relation $\Rs_{\Bisim}$ to derive a finite model bisimilar to $\Ss_{\MPETC}$ as follows: 
	\begin{defn}\label{def:bisim} The MPETC traffic model is the system $$\Ss^{\Bisim} \coloneqq \left(\Xs^{\Bisim}, \Xs^{\Bisim}, \emptyset, \Es^{\Sim}, \Ys_{\MPETC}, H^{\Sim}\right)$$ with $\Xs^{\Bisim} \coloneqq \pi_{\Rs_{\Bisim}}(\Xs)$.
	\end{defn}
	
	This model is a subset of $\Ss^{\Sim}$, generating only inter-sample sequences that can be produced by the concrete system $\Ss_{\MPETC}$. Topologically, it is still a tree, such as $\Ss^{\Sim}$, but with fewer states (see Figure \ref{fig:abstraction}). Our main result follows:
	
	\begin{figure}[tb]
		\begin{center}
			\begin{footnotesize}
			\begin{tikzpicture}[->,>=stealth',shorten >=1pt, auto, node distance=1.5cm,
			semithick, xscale=0.8, yscale=0.9]
			\tikzset{every state/.style={minimum size=2em, inner sep=2pt}}
			
			\node[state] 		 (0)    {$\varepsilon$};
			\node[state]         (1) at ([shift=({135:1.2 cm})]0)  {1};
			\node[state]         (2) at ([shift=({45:1.2 cm})]0) {2};
			\node[state]         (3) at ([shift=({-45:1.2 cm})]0) {3};
			\node[state]         (4) at ([shift=({-135:1.2 cm})]0)  {4};
			\node[state]         (11) at ([shift=({195:1.2 cm})]1)  {1,1};
			\node[state]         (21) at ([shift=({155:1.2 cm})]1)  {2,1};
			\node[state]         (31) at ([shift=({115:1.2 cm})]1)  {3,1};
			\node[state]         (41) at ([shift=({75:1.2 cm})]1)   {4,1};
			\node[state]         (12) at ([shift=({105:1.2 cm})]2)  {1,2};
			\node[state]         (22) at ([shift=({65:1.2 cm})]2)   {2,2};
			\node[state]         (32) at ([shift=({25:1.2 cm})]2)   {3,2};
			\node[state, dashed] (42) at ([shift=({-15:1.2 cm})]2)  {4,2};
			\node[state]         (13) at ([shift=({15:1.2 cm})]3)  {1,3};
			\node[state, dashed] (23) at ([shift=({-25:1.2 cm})]3)  {2,3};
			\node[state]         (33) at ([shift=({-65:1.2 cm})]3)  {3,3};
			\node[state, dashed] (43) at ([shift=({-105:1.2 cm})]3)   {4,3};
			\node[state, dashed] (14) at ([shift=({-75:1.2 cm})]4)  {1,4};
			\node[state, dashed] (24) at ([shift=({-115:1.2 cm})]4)   {2,4};
			\node[state, dashed] (34) at ([shift=({-155:1.2 cm})]4)   {3,4};
			\node[state]         (44) at ([shift=({-195:1.2 cm})]4)  {4,4};
			
			\path (11) edge (1) (1) edge (0)
			(21) edge (1)
			(31) edge (1)
			(41) edge (1)
			(12) edge (2) (2) edge (0)
			(22) edge (2)
			(32) edge (2)
			(42) edge (2)
			(13) edge (3) (3) edge (0)
			(23) edge (3)
			(33) edge (3)
			(43) edge (3)
			(14) edge (4) (4) edge (0)
			(24) edge (4)
			(34) edge (4)
			(44) edge (4)
			(0) edge [loop above] (0);
			\end{tikzpicture} \quad \quad \quad %
			\begin{tikzpicture}[->,>=stealth',shorten >=0pt, auto, node distance=1.0cm,
			semithick, xscale=1.0, yscale=1.0]
			\tikzset{every state/.style={minimum size=2em, inner sep=2pt}}
			
			\node[state] 		 (44)    {4,4};
			\node[state]		 (41) [right of=44] {4,1};
			\node[state]		 (11) [below right of=41] {1,1};
			\node[state]		 (12) [right of=11] {1,2};
			\node[state]		 (13) [below of=11] {1,3};
			\node[state]		 (21) [right of=13] {2,1};
			\node[state]		 (31) [left of=13] {3,1};
			\node[state]		 (33) [below of=13] {3,3};
			\node[state]		 (32) [right of=33] {3,2};
			\node[state]		 (22) [right of=21] {2,2};
			
			\path (44) edge [loop above] (44) (44) edge (41) 
			(41) edge (11) (41) edge [bend right] (13) (41) edge [bend left] (12)
			(11) edge (13) (11) edge (12) (11) edge [loop above] (11)
			(12) edge (21) (12) edge (22)
			(22) edge [loop above] (22) (22) edge (21)
			(31) edge (11) (31) edge (12) (31) edge [bend right] (13)
			(13) edge (31) (13) edge (33)
			(21) edge [bend left] (12) (21) edge (13) (21) edge (11)
			(33) edge (31) (33) edge [loop below] (33) (33) edge (32)
			(32) edge (21) (32) edge (22);			
			\end{tikzpicture}
			\end{footnotesize}
			\caption{\label{fig:abstraction} On the left, an illustration of $\Ss^{\Sim}$ (all states) and $\Ss^{\Bisim}$ (only solid-line states). On the right, a depiction of $\Ss^{\PETC}$ with $\Xs^{\PETC} = \{\sigma \in \Xs^{\Bisim}: |\sigma|=2\}$.}
			\vspace{-1.5em}
		\end{center}
	\end{figure}
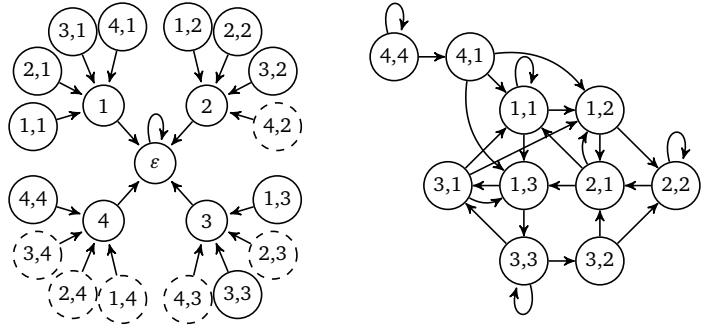
	
	\begin{thm}\label{thm:bisim}
		Let Assumptions \ref{assum:stable}--\ref{assum:compact} hold and $N \coloneqq \ceil{\log_a(r)}$. Then, $\Ss^{\MPETC} \approxeq \Ss^{\Bisim}$.
	\end{thm}
	\begin{proof}
		We show that $\Rs_{\Bisim}$ is a simulation relation from $\Ss_{\MPETC}$ to $\Ss^{\Bisim}$ and $\Rs^{-1}_{\Bisim}$ is a simulation relation from $\Ss^{\Bisim}$ to $\Ss_{\MPETC}$, checking each of the conditions of Definition \ref{def:sim}.
		
		{\bf Step 1:} $\Rs_{\Bisim}$ is a simulation relation from $\Ss_{\MPETC}$ to $\Ss^{\Bisim}$.
		
		For condition (i), take a point $\xv_0 \in \Xs_0 = \Xs$. It either belongs to $\Xs_{\periodic}$, for which Eq.~\eqref{eq:Xempty} provides its related state; or it takes $m$ PETC steps to reach $\Xs_{\periodic}$. In this latter case, it generates some trace $hk_1, hk_2, ..., hk_m, h^*, h^*, ...$ and therefore, by definition, it satisfies Eq.~\eqref{eq:sequence}. Hence, the related state $\dummy{k_1k_2...k_m}$ belongs to $\Xs^{\Bisim}.$
		Condition (ii) trivially holds by the definition of $\Rs_{\Bisim}$, and in particular by \eqref{eq:Xempty} and \eqref{eq:sequence_k1}.
		
		Finally, for condition (iii), take $(\xv, \sigma) \in \Rs_{\Bisim}.$ If $\xv \in \Xs_{\periodic}$, then $\sigma = \varepsilon$. From Assumption \ref{assum:periodicstable}, $\Post_{\Ss_{\MPETC}}(\xv) \in \Xs_{\periodic}$, which is related to $\varepsilon = \Post_{\Ss^\Bisim}(\varepsilon)$. If $\xv \notin \Xs_{\periodic}$, then $\sigma=k_1\sigma' \in \KN$. Therefore, $\Post_{\Ss_{\MPETC}}(\xv) = \Mm(k_1)\xv.$ From Assumption \ref{assum:stable}, $\Mm(k_1)\xv \in \Xs_0$; also, by inspecting Eq.~\eqref{eq:sequence}, $\Mm(k_1)\xv$ satisfies Eqs.~\eqref{eq:sequence_k2}--\eqref{eq:sequence_km} and Eqs.~\eqref{eq:sequence_x1_not_periodic}--\eqref{eq:sequence_xm-1_not_periodic}: this implies that $k_2...k_m = \sigma' = \Post_{\Ss^\Bisim}(k_1\sigma')$ is related to $\Mm(k_1)\xv$. 
		
		{\bf Step 2:} $\Rs^{-1}_{\Bisim}$ is a simulation relation from $\Ss^{\Bisim}$ to $\Ss_{\MPETC}$.
		
		For condition (i), if $\dummy{k_1k_2...k_m} \in \Xs_0^{\Bisim}$, then there exists a related initial state $\xv_0$ which satisfies Eq.~\eqref{eq:sequence}; hence, from Eq.~\eqref{eq:sequence_x0}, $\xv_0 \in \Xs_0$. For $\dummy\varepsilon$, any related state $\xv_0$ belongs to $\Xs_{\periodic} \subset \Xs = \Xs_0$.
		Condition (ii) is the same as in Step 1.
		Finally, condition (iii) is verified because the reasoning in Step 1 applies to every $\xv \in \Xs$ satisfying $(\sigma, \xv) \in \Rs^{-1}_{\Bisim}\!$.
	\end{proof}
	
	\begin{rem}\label{rem:decidable}
		Determining %
		if there exists $\xv$ satisfying Eq.~\eqref{eq:sequence} 
		 is a problem of checking non-emptiness of a semi-algebraic set, which has been proven to be decidable \cite{basu2006}. One tool that can be used to solve it is the SMT solver Z3 \cite{demoura2008z3}.
	\end{rem}
	
	\begin{prop}[Complexity]
		The state set $\Xs^{\Bisim}$ can be computed with $$\bigO\left(|K|^N (N|K|)^{\nx}\right)\cdot2^{\bigO(\nx)}$$ operations.
	\end{prop}
	
	\begin{proof}
		From Lemma \ref{lem:finite}, we have seen that there can be at most $$|K|(\left(|K|-1\right)^N-1)(|K|-1)^{-1}\in \bigO(|K|^N)$$ sampling sequences. Determining the state set $\Xs^{\Bisim}$ requires, in the worst case, checking the existence of all of those sequences. For a sequence of length $m$, Eq.~\eqref{eq:sequence} has one membership in $\Xs_0$ and $m$ memberships in $\Xs_{\periodic}$, each corresponding to one quadratic inequality; and $m$ memberships in $\Qs_k$, each corresponding to $k-\kmin+1$ quadratic inequalities. Therefore, in the worst case, Eq.~\eqref{eq:sequence} has $m+1+m|K|$ inequalities, or $1+N+N|K| \in \bigO(N|K|)$ for the longest sequence. The best known bound for deciding the existence of a real solution to a conjunction of $s$ polynomial inequalities of $\nx$ variables and maximum degree $d$ is $s^{\nx+1}d^{\bigO(\nx)}$  \cite{basu1996combinatorial}. Replacing $s$ by $1+N+N|K|$ and $d$ by 2, multiplying by the number of checks and working out the limits for big-O notation concludes the proof.
	\end{proof}
	
	\begin{rem}\label{rem:properalgorithm} While all sequences of length $N$ must be checked in the worst case, for other cases it is more efficient to employ a recursive algorithm, i.e., verifying Eq.~\eqref{eq:sequence} for sequences from length 1 to $N$. If a sequence $\sigma$ shorter than $N$ does not verify Eq.~\eqref{eq:sequence}, then no sequence $k\sigma$ can do. Hence, many checks can be discarded using this simple observation.
	\end{rem}
	
	\begin{rem}\label{rem:homo} Due to characteristics of the inequalities associated to Definition \ref{def:bisimrel},  one can set $V_0 = 1$ without loss of generality, with the only input to the model being the ratio of contraction $r$. For $V_0 = c > 0$, the model is the same: replace $\xv$ by $\sqrt{c}\xv$ in Eq.~\eqref{eq:sequence}, and $\sqrt{c}$ can be canceled out. \end{rem}
	
	\subsection{Derived results for the original PETC}
	
	With a few changes to $\Ss^{\Bisim}$, we can build a similar model of the PETC traffic that generates fewer spurious traces than, e.g., \cite{gleizer2020scalable}. This is because the PETC section of the MPETC trace is of course generated by the pure PETC system \eqref{eq:plant}--\eqref{eq:petc_time}. Hence, to simulate the PETC traffic, one could do the following: for a given state $\xv\in\R^\nx$, take $V_0 = V(\xv)$ and determine its related state $\dummy{k\sigma}$ from Eq.~\eqref{eq:sequence}. Now take its successor $\Mm(k)\xv$. Again, set $V_0 = V(\Mm(k)\xv)$ and determine its related state: it has to take the form $\dummy{\sigma\sigma'}$, i.e., its first inter-sample times should be all but the first inter-sample times of its predecessor.
	This idea is depicted in Fig.~\ref{fig:idea_PETC}.
	\begin{figure}
		\begin{center}
			\input{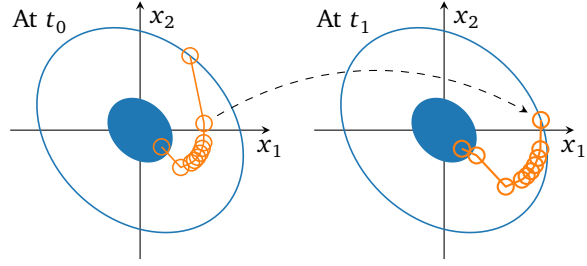}
			\caption{\label{fig:idea_PETC} Depiction of the strategy used to build the PETC traffic abstraction: the trajectory $\hat{\xiv}(t)$ is in orange with samples marked, and $\Xs_{\periodic}$ is the blue ellipse. The tail of the sequence from $t_0=0$ is the head of the sequence from the following sample $t_1$.}
			\vspace{-1.5em}
		\end{center}
	\end{figure}
	Let us formalize this procedure.
	\begin{defn}[PETC inter-sample sequence relation]\label{def:petcrel} Let $V_0 = 1.$ The relation $\Rs_{\PETC} \subseteq \R^{\nx} \times \KN$ is given by: $(\xv, \dummy{k_1k_2...k_m}) \in \Rs_{\PETC}$ iff $\xv/\sqrt{V(\xv)}$ satisfies Eq.~\eqref{eq:sequence}.
	\end{defn}
	\begin{thm}
		Let Assumption \ref{assum:stable} hold. Then, the system $$\Ss_{\PETC} \coloneqq (\Xs^{\PETC}, \Xs^{\PETC}, \emptyset, \Es^{\PETC}, \Ys_{\MPETC}, H^{\Sim}),$$ with $\Xs^{\PETC} \coloneqq \pi_{\Rs_{\PETC}}(\R^{\nx})$ and 
		$\Es^{\PETC} = \{(\dummy{k\sigma}, \dummy{\sigma\sigma'}) | \dummy{k\sigma}, \dummy{\sigma\sigma'} \in \Xs^{\PETC}\},$
		simulates the traffic generated by System \eqref{eq:plant}--\eqref{eq:quadtrig}.
	\end{thm}
	\begin{proof}
		\rev{Take an initial state $\xv\in\R^\nx$,} a PETC trajectory $\xiv_{\xv}(t)$ and its associated inter-sample sequence $k_1k_2...k_m$, after which $V(\xiv_{\xv}(t_m)) \leq rV(\xv)$ but $V(\xiv_{\xv}(t_{m-1})) > rV(\xv)$. This implies that $\xv/\sqrt{V(\xv)}$ satisfies Eq.~\eqref{eq:sequence}. Hence, $(\xv/\sqrt{V(\xv)}, \dummy{k_1k_2...k_m}) \in \Rs_{\PETC}$, and condition (i) of Def.~\ref{def:sim} holds. Condition (ii) is trivially satisfied, as $H(\xv) = hk_1 = H^{\Sim}(k_1k_2...k_m)$. For condition (iii), take its related sequence $k_1k_2...k_m$. The successor of $\xv$ is $\xv' \coloneqq \xiv_{\xv}(hk_1) = \Mm(k_1)\xv$, which satisfies Eqs.~\eqref{eq:sequence_k2}--\eqref{eq:sequence_km} and Eqs.~\eqref{eq:sequence_x1_not_periodic}--\eqref{eq:sequence_xm-1_not_periodic}; from homogeneity of $\Qs_k$, $\Mm(k_1)\xv/\sqrt{V(\Mm(k_1)\xv)}$ also satisfies Eqs.~\eqref{eq:sequence_k2}--\eqref{eq:sequence_km-1}. Additionally, because of Assumption \ref{assum:stable}, we have that $V(\xv') < V(\xv)$; hence, Eqs.~\eqref{eq:sequence_x1_not_periodic}--\eqref{eq:sequence_xm-1_not_periodic} holding for $\xv/\sqrt{V(\xv)}$ imply that $V(\xiv_{\xv}(t_i)) > rV(\xv) > rV(\xv')$ for all $1 \leq i \leq m$, and therefore Eqs.~\eqref{eq:sequence_x1_not_periodic}--\eqref{eq:sequence_xm-1_not_periodic} also hold for $\xv'/\sqrt{V(\xv')}$. This shows that the prefix of the sequence related to $\xv'$ is $k_2...k_m$. Finally, $\xv'/\sqrt{V(\xv')}$ satisfies Eq.~\eqref{eq:sequence} for some sequence in $\Ss_{\PETC}$; combining with the conclusion about the prefix above, $(\xv',\sigma\sigma') \in \Rs_{\PETC}$. The related transition exists because $(k\sigma, \sigma\sigma') \in \Es^{\PETC}$ for every $\sigma\sigma' \in \Xs^{\PETC}$.
	\end{proof}
	
	Note that $\Ss_{\PETC}$ is, in general, nondeterministic. A depiction of such construction is seen in Fig.~\ref{fig:abstraction}. Some useful verification applications can be derived from the model $\Ss_{\PETC}$:
	
	\begin{prop}\label{prop:freq} An upper bound for the average triggering frequency \rev{of system \eqref{eq:plant}, \eqref{eq:quadtrig}} is $f^* = \max_{\sigma\in\Xs^{\PETC}}(|\sigma|/(h{\sum_{k_i\in\sigma}k_i}))$. \end{prop}
	
	\begin{proof} In the worst case, $\Ss_{\PETC}$ generates $\sigma^* \coloneqq \arg\!\max_{\sigma\in\Xs^{\PETC}}(|\sigma|/(h{\sum_{k_i\in\sigma}k_i}))$ repeatedly.
	\end{proof}
	
	\begin{prop}\label{prop:ges} Let $T^* = h\max_{\sigma\in\Xs^{\PETC}}(\sum_{k_i\in\sigma}(k_i))$ be the longest (time-wise) sequence in $\Xs^{\PETC}$. Then $b^*\! =\! -\log(r)/2T^*$ is an upper bound for the GES decay rate of system \eqref{eq:plant},\eqref{eq:quadtrig}.\end{prop}
	\begin{proof}
		Take an initial state $\xv\in\R^{\nx}$, its related sequence $\sigma  \in \Xs^{\PETC}$, and set $T = h\sum_{k_i\in\sigma}(k_i)$. From Def.~\ref{def:petcrel}, $V(T) = V(\xiv_{\xv}(T)) \leq rV_0 = \e^{\log(r)}V(0) = \e^{-2b^*T}V(0).$ 
		From GES of the PETC (Assumption \ref{assum:stable}), $V(t) \leq M\e^{-2bt}V(0)$ for some $b > 0$ and $M < \infty$. Consider two cases. 
		
		{\bf Case 1:} $t<T$. Combining the inequalities above gives
		\begin{multline}\label{eq:t<T} 
		V(t) \leq M\e^{-2b(t-T)}V(T) \leq \e^{-2b^*T}M\e^{-2b(t-T)}V(0) \\ \leq M\e^{2bT}\e^{-2b^*T}V(0) \leq M'\e^{-2b^*t}V(0),
		\end{multline}
		where $M' \coloneqq M\e^{2bT^*} \geq M\e^{2bT}$. 
		
		{\bf Case 2:} $t>T$; then we can partition the trajectory $\xiv_{\xv}(t)$ in intervals $[0, t_{m_1}],$  $[t_{m_1}, t_{m_2}],$  $...,$ $[t_{m_n}, t]$ satisfying $V(t_{m_i}) \leq rV(t_{m_{i-1}})$ and $V(t_{m_i})  > rV(t_{m_{i-1}-1}).$ Each interval but the last is associated with a sequence $\sigma_i \in \Xs^{\PETC}$, and therefore its duration is $T_i \leq T^*$. Thus, with $t' = t - t_{m_n},$
		\begin{multline*}
		V(t) = V\left(\textstyle\sum_{i=1}^nT_i + t'\right) \leq r^nV(t') \\ \stackrel{{\text{Eq.~\eqref{eq:t<T}}}}{\leq} r^nM'\e^{-2b^*t'}V(0)
		= M'\e^{-2b^*\!nT^*}\e^{-2b^*t'}V(0) \\
		= M'\e^{-2b^*(nT^*+t')}V(0) \leq M'\e^{-2b^*\!t}V(0).
		\end{multline*}
		
		In the two cases, we have $V(t) \leq M'\!\e^{-2b^*\!t}V(0)$, which implies $$|\xiv(t)| \leq \sqrt{M'\lambda_{\max}(\Pm)/\lambda_{\min}(\Pm)}\e^{-b^*\!t}|\xiv(0)|.$$
	\end{proof}
	
	We conjecture that Proposition \ref{prop:ges} provides a better estimate of the convergence rate of System \eqref{eq:plant}--\eqref{eq:quadtrig} than what can be obtained by, e.g., the theorems in \cite{heemels2013periodic}. The reason behind this conjecture is that, as $N\to\infty$ (or $r\to 0$), our bound approaches what would be the joint spectral radius of the associated discrete-time system (see, e.g.,~\cite{ahmadi2014joint}).

	\section{Numerical results}\label{sec:num}
	
	Consider a plant and controller of the form \eqref{eq:plant} from \cite{tabuada2007event}, and the Lyapunov function $V(\xv) = \xv\tran\Pl\xv$ such that the continuous-time closed-loop system satisfies $\d V(\xiv(t))/\d t = -\xiv(t)\tran\Ql\xiv(t)$, determined by the following matrices:
	\begin{gather*} \Am = \begin{bmatrix}0 & 1 \\ -2 & 3\end{bmatrix}, \ \Bm = \begin{bmatrix}0 \\ 1\end{bmatrix}, \ \Km = \begin{bmatrix}1 & -4\end{bmatrix}, \\
	\Pl = \begin{bmatrix}1 & 0.25 \\ 0.25 & 1\end{bmatrix}, \ \Ql = \begin{bmatrix}0.5 & 0.25 \\ 0.25 & 1.5\end{bmatrix}.
	\end{gather*}
    For the PETC implementation, we use a predictive Lyapunov-based triggering condition of the form
	$ V(\zetav(t)) > -\rho \zetav(t)\tran\Ql\zetav(t), $
	where $\zetav(t) \coloneqq \Am_\d(1)\xiv(t) + \Bm_\d(1)\Km\hat{\xiv}(t)$ is the next-sample prediction of the state and $0 < \rho < 1$ is the triggering parameter, here set to $\rho=0.8$. Setting $h = 0.1$ and $\bar{k} = 6$, we put it in the form \eqref{eq:quadtrig}, and obtained $a=0.952$ using LMIs based on Remark \ref{rem:a_computation}. %
	For the periodic \rev{phase}, the maximum period that verifies Assumption \ref{assum:periodicstable} is $h_{\periodic} = 0.4$ (with resolution of 0.01). Finally, we verified that $\underline{k}=1$ and set $r=0.1$.
	\begin{figure}
		\begin{center}
			\input{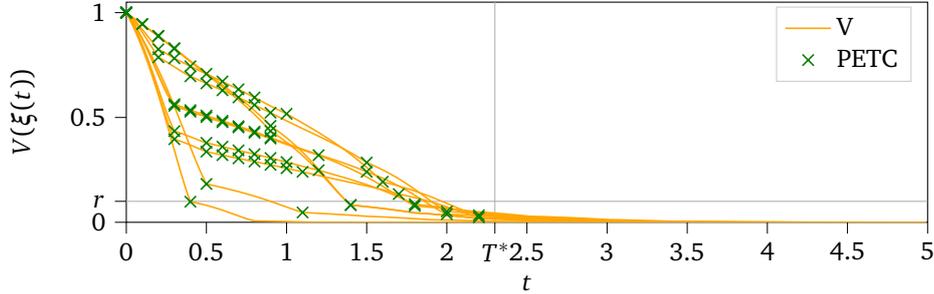}
			\vspace{-2em}
			\caption{\label{fig:lyap} {Trajectory of the Lyapunov function for 10 different initial conditions under MPETC, with PETC samples marked. The maximum time $T^*$ it takes to reach $\Xs_{\periodic} = r\Xs_0$ is highlighted.} }
			\vspace{-1.5em}
		\end{center}
	\end{figure}
	Lemma \ref{lem:finite} gives $N=47$ and a worst-case value of $8.5\cdot10^{32}$ bisimulation states. We computed the bisimilar state set $\Xs^{\Bisim}$ implementing a recursive algorithm (as discussed in Remark \ref{rem:properalgorithm}), obtaining a total of 219 states, out of which 109 belong to $\Xs^{\PETC}$. The Python implementation, using Z3 to solve Eq.~\eqref{eq:sequence}, took 21 min to generate these state sets. 
	
	Using MPETC, the maximum time it takes for $\Xs_{\periodic}$ to be reached is $$T^* = h\max_{\sigma\in\Xs^{\Bisim}}(\sum_{k_i\in\sigma}(k_i)) = 2.3.$$ This is highlighted in Fig.~\ref{fig:lyap}, which shows simulations from 10 different initial conditions.
	For PETC, applying Proposition \ref{prop:ges} gives $b^* = 0.5,$ while the best GES rate that can be obtained using the LMI approaches from \cite{heemels2013periodic} is $b = 0.23$, using Theorem III.4. For the average PETC sampling frequency, Proposition \ref{prop:freq} gives $f^* = 20/3$ (compared to $1/h = 10$), corresponding to the sequence $\sigma = (4,1,1,1,1,1).$
	
	\section{Conclusions}
	
	We have presented a practical alternative to ETC, the MPETC, \rev{which provides the benefits of PETC during transients and the traffic predictability of periodic sampling when close to steady state. Furthermore, we have presented a method to compute a bisimilar traffic model for MPETC.} In addition, we have presented some verification applications of the (bi)similar models that can be used for both PETC and MPETC. This is an important step towards understanding traffic characteristics of ETC, and it may support its applicability in real NCSs, since the traffic benefits are among the main motivations for the usage of ETC. Future work shall focus on expanding these models to scheduling \cite{mazo2018abstracted, gleizer2020scalable} and other triggering conditions, as well as using efficient relaxations to solve the satisfiability of Eq.~\eqref{eq:sequence}, such as $\delta$-SMT \cite{gao2013dreal} and semi-definite relaxations.
	\bibliographystyle{ieeetr} 
	\bibliography{mybib} 
		
\end{document}